\documentclass[hidelinks,onefignum,onetabnum]{dependencies/siamart220329}

\usepackage{lipsum}
\usepackage{amsfonts}
\usepackage{graphicx}
\usepackage{epstopdf}
\usepackage{mathtools}
\usepackage{algorithmic}
\ifpdf
  \DeclareGraphicsExtensions{.eps,.pdf,.png,.jpg}
\else
  \DeclareGraphicsExtensions{.eps}
\fi
\usepackage{amssymb}
\usepackage{mathabx}
\usepackage{pifont}

\usepackage[english]{babel}
\usepackage[T1]{fontenc}
\usepackage[utf8]{inputenc}

\usepackage{booktabs}

\newcommand{\Cmbb}{\mathbb{C}}

\newcommand{\Ocal}{{\cal O}}

\newcommand{\C}[1]{\mathbb{C}^{#1}}

\newcommand{\R}[1]{\mathbb{R}^{#1}}

\DeclareMathOperator{\trace}    {tr}

\newcommand{\bra}[1]{\left\langle{#1}\right|}
\newcommand{\ket}[1]{\left|{#1}\right\rangle}

\newsiamremark{remark}{Remark}
\newsiamremark{hypothesis}{Hypothesis}
\crefname{hypothesis}{Hypothesis}{Hypotheses}
\newsiamthm{claim}{Claim}

\headers{Analysis of iterative refinement for QLSA}{Harkness et al.}

\title{An analysis of iterative refinement for
quantum linear system solvers
}

\author{Adrian Harkness\thanks{Physical and Computational Sciences Directorate, Pacific Northwest National Laboratory, Richland, WA, USA, and Department of Industrial and Systems Engineering, Lehigh University, Bethlehem, PA, USA (\email adh323@lehigh.edu)}
\and Mohammadhossein Mohammadisiahroudi\thanks{Department of Mathematics and
Statistics, Quantum Science Institute, University of Maryland, Baltimore County,
Baltimore, MD, USA}
\and Brandon Augustino\thanks{Department of Industrial and Systems Engineering, Lehigh University, Bethlehem, PA, USA}
\and Giacomo Nannicini\thanks{Department of Industrial and Systems Engineering,
and Department of Electrical and Computer Engineering, University of Southern
California, Los Angeles, CA, USA}
\and Tam\'as Terlaky\footnotemark[3]}
\usepackage{amsopn}

\makeatletter
\def\refstepcounter@optarg[#1]#2{%
  \cref@old@refstepcounter{#2}%
  \cref@constructprefix{#2}{\cref@result}%
  \@ifundefined{cref@#1@alias}%
    {\def\@tempa{#1}}%
    {\def\@tempa{\csname cref@#1@alias\endcsname}}%
  \protected@edef\cref@currentlabel{%
    [\@tempa][\arabic{#2}][\cref@result]%
    \csname p@#2\endcsname\csname the#2\endcsname}%
}
\makeatother
\ifpdf
\hypersetup{
  pdftitle={IR-LSP},
  pdfauthor={M.Mohammadi}
}
\fi

\begin{document}

\maketitle

\begin{abstract}
We present and analyze an iterative refinement (IR) framework for improving the precision dependence of algorithms that combine a quantum linear system algorithm (QLSA) with quantum state tomography. 
Existing QLSAs achieve polylogarithmic dependence on the inverse error tolerance to prepare a quantum state encoding the solution, but extracting a classical description of such a state via tomography typically introduces a polynomial dependence on the target precision.
To retain polylogarithmic dependence in the inverse error tolerance throughout the entire process, we develop an IR scheme that solves a sequence of related linear systems to progressively refine the solution while requiring only fixed-precision quantum subroutines, finally obtaining a classical description of a high-precision solution. 
%
%
We analyze the proposed framework under three input models: quantum-read/classical-write RAM (QRAM), linear combinations of unitaries (LCU), and sparse-access oracles. 
We evaluate the proposed scheme through numerical experiments on both quantum simulators and real quantum hardware. 
The results demonstrate that iterative refinement efficiently improves the precision of the solution and exhibits robustness to hardware noise.
\end{abstract}



\section{Introduction}

It is widely believed that quantum computers can solve some problems faster than classical algorithms (``classical'' means ``non-quantum''). 
%
The increasing size and capabilities of current devices  have added further momentum to the development of theory and algorithms for quantum computing, with the overarching goal of providing convincing evidence of a quantum advantage.

Recently, considerable attention has been devoted to quantum algorithms for linear algebraic tasks. A particular interest lies in a class of methods that use quantum algorithms to prepare a quantum state $\ket{x}$ that is proportional to the solution of the linear system 
\begin{equation}\label{e:LS}
    Ax = b, 
\end{equation}
for a given $s$-row sparse matrix $A \in \mathbb{R}^{d \times d}$ (i.e., $A$ has at most $s$ nonzero entries per row) and a vector $b \in \R{d}$. In the \textit{quantum linear system problem} (QLSP), one has access to some quantum circuit describing $A$ and the ability to prepare a quantum state $\ket{b}$ that is proportional to the right-hand side vector $b$. 

Research into this subfield began with the work of Harrow, Hassidim, and Lloyd \cite{harrow2009quantum}, who proposed what has come to be known as the HHL algorithm for solving the QLSP. In this seminal work, it was shown that a quantum computer could be used to solve a QLSP with a worst-case complexity\footnote{We use $\widetilde{\Ocal}$, which suppresses the polylogarithmic factors in the ``Big-O'' notation.
Subscripts of $\widetilde{\Ocal}$ indicate the parameters/quantities occurring in the suppressed polylogarithmic factors.} of
$$\widetilde{\Ocal}_{d} \left(\frac{s^2\kappa_A^2}{\epsilon} \right),$$
where $\epsilon > 0$ is the accuracy to which the solution is obtained. 
Although the HHL algorithm originally exhibited quadratic dependence on $\kappa_A$, its poly-logarithmic dependence on the dimension of the problem opened the doors for a potential exponential quantum speedup for solving QLSPs. 
Given the fundamental role that linear systems play in science and mathematics, the result in \cite{harrow2009quantum} led to a series of works focused on improving the performance of QLSAs, which we review next.

\subsection{Literature Review}\label{ss: LR}
The first enhancements to the complexity of the HHL algorithm were made in \cite{ambainis2012variable, childs2017quantum, subacsi2019quantum}, in which the dependence on $\kappa_A$ was improved from quadratic to almost linear (i.e., $\Ocal \left(\kappa_A \log \kappa_A \right)$). Further improvements came by developing better tools for Hamiltonian simulation. Indeed, the HHL algorithm is built on the idea of using Hamiltonian simulation to simulate the unitary $e^{-i Ht}$ on an input state $| \psi \rangle$ for some time $t$, and early work that extended the HHL algorithm relied primarily on the \textit{sparse-access input model} to query the elements of $H$. It was shown by \cite{low2019hamiltonianOpt} that one could perform optimal Hamiltonian simulation when the input matrix is given as a \textit{block-encoding}. 
Under this framework, one has access to unitary operators that have the coefficient matrix in their top-left block:
$$ U = \begin{pmatrix} A/\alpha & \cdot \\
\cdot & \cdot \end{pmatrix},$$
where $\alpha \geq 1$ is a normalization factor (recall that $U$ has operator norm 1, so normalization is necessary for general $A$). QLSAs based on the block-encoding framework were presented in \cite{vazquez2022enhancing, childs2017quantum, wossnig2018Quantum}, and the running time of these QLSAs was subsequently improved by \cite{chakraborty2018power} using variable-time amplitude  amplification \cite{ambainis2012variable} and estimation. When provided access to a quantum-read/classical-write RAM (QRAM), \cite{chakraborty2018power} obtains a QLSA with worst-case query complexity
$$\widetilde{\Ocal}_{d, \kappa_A, \frac{1}{\epsilon}}\left(\kappa_A \alpha \right).$$ 
It is additionally shown in~\cite{chakraborty2018power} how one can prepare block-encodings of matrices using various input models. We point out that using techniques based on quantum singular value transformation \cite{gilyen2019quantum}, one can design a QLSA with improved dependence on polylogarithmic factors compared to the algorithm in \cite{chakraborty2018power}. 

More recent developments have established QLSAs with optimal asymptotic dependence on both the condition number and the solution accuracy. 
Costa et al.~\cite{costa2022optimal} introduced an optimal-scaling QLSA based on a discrete adiabatic theorem and qubit\-ized quantum walks, achieving query complexity
\begin{equation*}
    \widetilde{\Ocal}_{\frac{1}{\epsilon}}\left(\kappa_A \right),
\end{equation*}
which is optimal in its combined dependence on $\kappa_A$ and $\epsilon$. 
This result initiated a more recent line of research aimed not only at attaining optimal asymptotic scaling, but also at reducing the constant factors hidden in the complexity bounds.
Dalzell~\cite{dalzell2024shortcut} proposed the \textit{shortcut} method, that uses quantum singular value transformation and eigenstate filtering techniques to obtain the same optimal $\Ocal(\kappa_A\log(1/\epsilon))$ scaling while providing substantially improved explicit constant-factor guarantees.
Costa et al.~\cite{costa2026constant} performed a detailed numerical comparison of optimal QLSAs, including the discrete-adiabatic quantum-walk and shortcut approaches, demonstrating that their practical performance can differ substantially despite having identical asymptotic complexity.
These developments suggest that the theoretical complexity of QLSAs has reached asymptotically optimal dependence on $\kappa_A$ and $\epsilon$, shifting increasing attention toward constant factors and the practical cost of employing QLSAs as subroutines in larger quantum algorithms.
These developments have opened a line of research investigating how QLSAs can be leveraged to design quantum algorithms that obtain speedups over classical algorithms in which the dominant operation is solving a linear system of equations; examples are \cite{mohammadisiahroudi2024efficient, mohammadisiahroudi2023inexact, augustino2023quantum}.
%

Besides QLSAs, the other main building block of the approach studied in this paper is iterative refinement (IR). IR was originally developed at a time when classical computers possessed capabilities just beyond the reach of today's quantum computers, and its initial conception was a framework for computing extended-precision solutions to linear systems of equations using low-precision arithmetic \cite{golub2013matrix, wilkinson1994rounding}. 
These ideas have since been extended to compute highly accurate (and at times, \textit{exact}) solutions to linear \cite{gleixner2020linear,gleixner2012improving, gleixner2016iterative, mohammadisiahroudi2025improvements, wu2023inexact} and semidefinite optimization problems \cite{mohammadisiahroudi2025quantum}, and recent works have investigated their adaptability to quantum algorithms. 
IR has recently been deployed to improve the practical and theoretical performance of inexact classical algorithms for solving linear systems \cite{carson2017new, carson2018accelerating}.

\begin{table*}[t]
\centering
\caption{Comparison of representative quantum linear system algorithms. Here,
$\kappa$ denotes the condition number, $s$ the sparsity of the matrix, and $\epsilon$ the target precision.}
\label{tab:QLSA_comparison}
\renewcommand{\arraystretch}{1.15}
\begin{tabular}{p{1.1cm}p{3.8cm}p{3.5cm}p{2.8cm}}
\toprule
\textbf{Paper} &
\textbf{Technique} &
\textbf{Query Complexity} &
\textbf{Input Model} \\
\midrule

\cite{harrow2009quantum}
&
Phase estimation
&
$\widetilde{\mathcal O}_{d}\!\left(\kappa^2 s^2/\epsilon\right)$
&
Sparse access
\\

\cite{ambainis2012variable}
&
Variable-time amplitude amplification
&
$\widetilde{\mathcal O}_{d}\!\left(\kappa s/\epsilon\right)$
&
Sparse access
\\

\cite{childs2017quantum}
&
Fourier/Chebyshev approximation
&
$\widetilde{\mathcal O}_{d,\kappa/\epsilon}\!\left(\kappa s\right)$
&
Sparse access
\\

\cite{chakraborty2018power}
&
QSVT
&
$\widetilde{\mathcal O}_{d,\kappa/\epsilon}\!\left(\alpha\,\kappa\right)$
&
Block encoding
\\

\cite{an2019quantum}
&
Time-optimal adiabatic quantum computing
&
$\widetilde{\mathcal O}_{d, \kappa/\epsilon}\!\left(s\kappa\right)$
&
Input oracle access
\\
\cite{dalzell2024shortcut}
&
Optimal QLSA + QSVT
&
$\widetilde{\mathcal O}_{d,1/\epsilon}\!\left(\alpha \,\kappa\right)$
&
Block encoding
\\

\bottomrule
\end{tabular}
\end{table*}

\subsection{Contributions}\label{ss: CT}

We develop an iterative refinement (IR) framework for obtaining a classical solution to the linear system problem \eqref{e:LS} using a quantum linear system algorithm (QLSA) followed by quantum state tomography (QTA). The main idea is to use QLSA and tomography only at a fixed precision in every refinement step, while progressively improving the accuracy of the classical solution through residual correction. As a consequence, the dependence of the overall algorithm on the final target precision is reduced from polynomial to polylogarithmic.

We analyze the proposed framework under three input models: QRAM, sparse-access oracles, and linear combinations of unitaries (LCU). In all three cases, Algorithm~\ref{LS:alg:iterative refinement for LS} requires only
$
\Ocal\!\left(
\log\!\left(\frac{\|b\|\kappa_A}{\epsilon}\right)
\right)
$
iterations to obtain a classical vector $x$ satisfying
$
\|x-A^{-1}b\|\leq\epsilon.
$
Since every quantum linear-system solve and tomography step is performed at a fixed precision $\xi<1$, the dependence on the final target precision $\epsilon$ appears only through polylogarithmic factors.

In the QRAM input model, for an $s$-sparse matrix $A$, our algorithm obtains an $\epsilon$-precise classical solution using

$$
\widetilde{\Ocal}_{d,\kappa_A,\frac{\|b\|}{\epsilon}}
\!\left(
\|A\|_F d\kappa_A^2
\right)
$$

accesses to QRAM and

$$
\widetilde{\Ocal}_{\kappa_A,\frac{\|b\|}{\epsilon}}
\!\left(
ds
\right)
$$

classical arithmetic operations.

Under the sparse-access input model, the algorithm requires

$$
\widetilde{\Ocal}_{d,\kappa_A,\frac{\|b\|}{\epsilon}}
\!\left(
\sqrt{s}d\,\kappa_A^2
\right)
$$

queries to the sparse-access oracles describing $A$, together with

$$
\widetilde{\Ocal}_{d,\kappa_A,\frac{\|b\|}{\epsilon}}
\!\left(
\sqrt{s}d^2\kappa_A^2
\right)
$$

additional quantum gates and

$$
\widetilde{\Ocal}_{\kappa_A,\frac{\|b\|}{\epsilon}}
\!\left(
ds
\right)
$$

classical arithmetic operations.

Finally, when the coefficient matrix is represented as a linear combination of unitaries
$
A=\sum_{i=0}^{K-1}\alpha_iU_i,
$
and a classical multiplication $U_i v$ can be performed using at most $C_U$ arithmetic operations, our algorithm requires
$$
\widetilde{\Ocal}_{d,\kappa_A,\frac{\|b\|}{\epsilon}}
\!\left(
\alpha d\kappa_A^2
\right)
$$
calls to the LCU input oracles,
$$
\widetilde{\Ocal}_{d,\alpha,\kappa_A,\frac{\|b\|}{\epsilon}}
\!\left(
\alpha d^2\kappa_A^2
\right)
$$
additional quantum gates, and
$$ \tilde{\Ocal}_{\kappa_A, \frac{ \| b \|}{\epsilon}} \left(K C_U\right)$$
classical arithmetic operations.

Thus, irrespective of the input model, iterative refinement removes the polynomial dependence on the final inverse precision that arises when a QLSA is followed directly by tomography. In particular, rather than increasing the precision of the quantum subroutines as $\epsilon$ decreases, our method executes only fixed-precision quantum circuits and achieves high accuracy through a logarithmic number of classical refinement steps.

We complement the theoretical analysis with numerical experiments using both quantum simulators and real quantum hardware. The experiments demonstrate geometric residual reduction in the noiseless setting and continued convergence in the presence of hardware noise. They also illustrate the principal resource advantage of the IR framework: high-precision classical solutions can be obtained through repeated execution of relatively shallow, low-precision quantum circuits instead of a single high-precision circuit whose depth grows rapidly with the desired accuracy.

This paper is an updated version of an earlier unpublished technical report \cite{mohammadisiaroudi2023exponentially}. Gily\'en et al.~\cite{gilyen24iterative} subsequently improved the dependence on $\kappa_A$ to linear by decomposing the solution into different eigenspaces and applying iterative refinement separately within those eigenspaces. To the best of our knowledge, the results in \cite{gilyen24iterative} are not available in the open literature or on online repositories, and therefore we do not include them in our quantitative comparison. Related ideas are discussed in \cite{chen2024quantum}.

The remainder of this paper is organized as follows. Section~\ref{s:prelim} introduces notation and assumptions. Section~\ref{s:qla} reviews the quantum linear-algebra tools, input models, QLSAs, and tomography algorithms used in our analysis. Section~\ref{sec:IR} presents the proposed iterative refinement algorithm and establishes its convergence. We then analyze its complexity under the QRAM, sparse-access, and LCU input models in Section~\ref{sec:complexity_analysis_input_models}, followed by numerical experiments on simulators and quantum hardware in Section~\ref{sec:num}. Section~\ref{s:con} concludes the paper.

\section{Notation and assumptions}\label{s:prelim}
For any integer $n$, let $[n] = \{0, \dots, n-1\}$. The $i$-th element of a vector $x \in \R{n}$ is denoted by $x_i$ for $i \in [n]$, and we denote the $ij$-th element of a matrix $M \in \R{m \times n}$ by $M_{ij}$ for $i \in [m]$ and $j \in [n]$. To refer to the $i$-th row of a matrix $M$, we write $M_{i, \cdot}$ and write $M_{ \cdot, j}$ when referring to its $j$-th column. Note that we use zero-based indices for vectors and matrices (i.e., the first element has index zero) because this considerably simplifies notation when dealing with coefficients of quantum states and the corresponding basis states. For a complex-valued vector $x \in \Cmbb^d$, its real component is $\Re(x) \in \R{d}$.
For $x \in \R{n}$, its amplitude encoding $\ket{x}$, is the $(\log n)$-qubit state 
$$ \ket{x} = \frac{1}{\|x\|} \sum_{j \in [n]} x_j \ket{j}.$$ 
We assume that all logarithms are base 2, and that the sizes of all spaces are powers of 2; this is without loss of generality as we can always pad vectors and matrices with zeroes.
The smallest and largest singular values of a matrix $A$ are denoted by $\sigma_{\min}(A)$ and $\sigma_{\max}(A)$, and the smallest and largest eigenvalues are denoted by $\lambda_{\min}(A),$ and $\lambda_{\max}(A)$. 
Unless otherwise specified, sparsity of matrices is intended column- and row-wise: a matrix $A$ is $s$-sparse if each column and row of $A$ has at most $s$ nonzero elements. The requirement of simultaneous column and row sparsity comes from the fact that for a general matrix $A$, before applying a QLSA we typically need to embed $A$ into a larger Hermitian matrix; column and row $s$-sparsity of $A$ ensures that the rows of the resulting matrix are $s$-sparse, which is assumed by the Hamiltonian simulation algorithms used in our QLSA of choice. For the quantum linear system solver we assume that $\sigma_{\max}(A) = 1$ and $\sigma_{\min}(A) = 1/\kappa_A$, as is common in QLSA literature. Up to rescaling, this is equivalent to assuming that we know the smallest and largest singular value of $A$. We keep this assumption throughout the paper, noting that it is sufficient to know singular values up to a constant prefactor.

All quantum algorithms discussed below are successful with probability at least $1-\delta$, with a cost that scales as $\Ocal(\operatorname{polylog}\frac{1}{\delta})$. Because this affects the complexity analysis only in polylogarithmic terms, to ease the notation we do not report this cost explicitly and state all results as if the algorithms were always successful: the dependence on $\delta$ should be implicitly assumed.

\section{Preliminaries on quantum linear algebra}
\label{s:qla}
The (classical) LSP is formally defined as follows.

\begin{definition}[LSP]\label{def: LSP}
Let $A \in \R{d \times p}$ be a matrix with condition number $\kappa_A$. Given $b\in \R{d}$, the linear system problem (LSP) is to determine $x\in\mathbb{R}^{p}$ such that $$Ax=b.$$ 
\end{definition}

When $d=p$ and $A$ is nonsingular, the LSP has the unique solution $x_*=A^{-1}b$. Classical direct methods such as Gaussian elimination, and factorization techniques require $\Ocal(d^3)$ arithmetic operations in the dense setting. Faster asymptotic algorithms are possible through fast matrix multiplication. Strassen's method reduces this cost to $\Ocal(d^{2.81})$ \cite{strassen1969gaussian}, while subsequent work has established matrix-multiplication exponents below $2.373$ \cite{coppersmith1987matrix, le2014powers}. For sufficiently sparse matrices, improved bounds are also possible; for example, \cite{peng2021solving} obtains a complexity of $\Ocal(d^{2.331645})$ under additional assumptions on sparsity and conditioning. Although in theory this has brought the computational complexity of matrix multiplication closer to $d^2$ than to $d^3$, these fast direct methods are typically not used in practice.

Inexact iterative methods can provide better dependence on the dimension by directly computing an approximation to $A^{-1}b$. The conjugate gradient (CG) method \cite{hestenes1952methods} solves symmetric positive definite systems with worst-case complexity
\[
    \Ocal\!\left(
        ds\sqrt{\kappa_A}\log(1/\epsilon)
    \right).
\]
For indefinite or nonsymmetric systems, related Krylov methods such as MINRES and GMRES are commonly used. Alternatively, one can symmetrize the system as
\[
    A^\top A x=A^\top b,
\]
although this squares the condition number and can increase the computational cost.

Polynomial approximation methods provide another route to solving linear systems. In particular, Chebyshev polynomials can approximate the inverse function on the spectrum of the coefficient matrix with degree
\[
    \operatorname{deg}(q)
    =
    \Ocal\!\left(
        \kappa_A\log(1/\epsilon)
    \right),
\]
leading to a comparable dependence on $\kappa_A$ and $\epsilon$ through repeated matrix-vector products; see, e.g., \cite[Section 6.11]{saad2003iterative}. For a more detailed discussion of classical iterative methods for linear systems, we refer the reader to \cite{saad2003iterative}.

We define the quantum analogue of the LSP following
\cite{dalzell2024shortcut}; see also
\cite{ambainis2012variable, chakraborty2018power, harrow2009quantum,
subacsi2019quantum} for earlier formulations.
\begin{definition}[QLSP]\label{def:QLSP}
Let $A \in \mathbb{R}^{d \times d}$ be a matrix whose nonzero singular values lie
in the interval $[\frac{1}{\kappa_A}, 1]$, and let $b \in \mathbb{R}^d$ be a vector
in the column space of $A$. Let $x \in \mathbb{R}^d$ denote the solution of minimum
norm to $Ax = b$, and let $\ket{b}$ and $\ket{x}$ be the amplitude encodings of $b$ and $x$, respectively.
Given $\epsilon \in (0,1)$, the Quantum Linear System Problem (QLSP) is to
find a (possibly mixed) quantum state $\rho$ satisfying
\begin{equation}\label{e:somma2}
    \frac{1}{2} \left\| \rho - \ket{x}\!\bra{x} \right\|_{\trace} \leq \epsilon,
\end{equation}
where $\| X \|_{\trace} = \trace\big( \sqrt{XX^{\dagger}} \big)$.
\end{definition}
Here, $\kappa_A$ is an upper bound on the condition number of $A$ restricted
to the orthogonal complement of its kernel; when $A$ is invertible, this is
the usual condition number, and the solution state reduces to
\begin{equation}\label{e:quantumSolutionState}
    \ket{x} = \frac{A^{-1} \ket{b}}{\left\| A^{-1} \ket{b} \right\|}.
\end{equation}

In order to efficiently prepare a state satisfying \eqref{e:somma2}, QLSAs
require access to $A$ and $b$ via some appropriate quantum data structure. As
we discuss next, the way in which the input is presented impacts the overall
running time of the QLSA in a number of ways, such as the cost of preparing a
state encoding the right-hand side and the subnormalization factor $\alpha$ of
the block-encoding of the coefficient matrix.

\subsection{Input models}
We conduct a worst-case overall running time analysis for our algorithm using the standard gate-based quantum circuit model, as well as a model in which our data is stored using quantum-read/classical-write RAM (QRAM). 
The impact of the assumed input model is not limited to the cost associated with preparing quantum linear systems of equations. 
Instead, in the context of our IR scheme, whether or not we have access to QRAM also influences the QLSA we use and the cost of performing quantum state tomography. 
To study the impact of the data structure on our algorithms' overall complexity, we consider three input models: the \textit{sparse access input model}, the \textit{QRAM input model}, and the \textit{linear combination of unitaries input model}. 

\subsubsection{Sparse access model}
\label{s:sparse}
The sparse access input model directly corresponds to the commonly used way of accessing sparse matrices in classical computation: it assumes access to the nonzero elements of the $s$-sparse matrix $A \in \R{d \times d}$. 

As in \cite{berry2015hamiltonian}, access to $A$ is provided by two oracles. 
The first of these two oracles is a unitary $O^{\text{sparse}}_A$ that calculates the function $\operatorname{index}: [n] \times [s] \to [n]$ for input $(i,\ell)$, giving the position of the $\ell$-th nonzero element in the $i$-th row of $A$:
$$ O^{\text{sparse}}_A \ket{i,\ell} = \ket{i, \operatorname{index}(i,\ell)}.$$
The second oracle is the unitary $O_A \ket{i, j, z} = \ket{i,j,z \oplus A_{ij}}$ that outputs a binary representation of the entries of $A$ for every $i,j \in [d]$.

\subsubsection{Quantum random access memory}
\label{s:qram}
We consider the quantum-accessi\-ble data structure presented in \cite{kerenidis2017quantum}. 
Quantum-accessible refers to the fact that it is a classical data structure accessible in coherent superposition when stored in QRAM. This allows us to implement certain quantum operations more efficiently than in the sparse-access input model without QRAM. 
Note that we only need classical write access to QRAM; we do not need to be able to write in a superposition.
As in \cite{chakraborty2018power}, our analysis relies on the standard assumption that a QRAM of size $w$ can be accessed at cost $\operatorname{polylog} (w)$, and that the $\Ocal (w)$ many gates required to implement the QRAM can be arranged in parallel to achieve a circuit depth of order $\Ocal \left(\operatorname{polylog} (w)\right)$, see \cite[Section 2.2]{chakraborty2018power}. Computational complexity of algorithms using QRAM is not the total number of gates, but a measure of time that relates to the depth of the QRAM circuit. To be precise, we report the computational complexity of such algorithms in terms of number of QRAM accesses, number of additional gates, and number of classical operations to initialize the QRAM, see, e.g., \cite[Ch.~5]{nannicini2025quantum}. In the QRAM input model, all classical data (i.e., $A$ and $b$, and any other r.h.s.\ vector that arises in the course of the IR scheme) is encoded in appropriate data structures stored in QRAM. The data structures can be efficiently computed and they are described in \cite{kerenidis2017quantum,chakraborty2018power}; the details are irrelevant for the purposes of this paper.
\subsubsection{Linear Combination of Unitaries}
\label{s:lcu}
A third input model commonly used to access a matrix $A \in \C{n \times n}$, first proposed in \cite{childs2012HamSimLCU}, assumes that it is provided as a linear combination of unitary matrices (LCU):
$$
A = \sum_{i=0}^{K-1} \alpha_i U_i.
$$
A common unitary basis to use, particularly when A is a Hamiltonian, is the Pauli basis: $U_i \in \{c \cdot P_1 \otimes P_2 \otimes ... \otimes P_n: c\in \{\pm1, \pm i\}, P_j \in \{I, X, Y, Z\} \}$.
Given any linear operator $A$ on $\log(n)$ qubits, a representation of A in the Pauli basis always exists, thus there always exists an LCU representation for $A$ (the sparsity of the representation affects the computational cost of solving linear systems in this input model). 

Next, we provide results from \cite{gilyen2019thesis, gilyen2019quantum, dalzell2024shortcut}, which establish how one can efficiently construct block-encodings of matrices using the sparse-access, QRAM, and LCU input models and use these block-encodings to solve linear systems with a quantum computer. 

\subsection{Complexity of the QLSP in the three input models}
In the three input models discussed above, the QLSA with the best asymptotic query complexity uses the block-encoding model. We define here the necessary notions. 
%
\begin{definition}[Block-encoding] \label{prop:qramblockenc}
Let $M \in \Cmbb^{2^n \times 2^n}$ be an $n$-qubit operator and $\epsilon > 0$. 
Then, an $(n + a)$-qubit unitary $U$ is an $(\alpha, a, \epsilon)$-block-encoding of $M$ if 
$ U = \begin{pmatrix}
\tilde{M} & \cdot \\
\cdot & \cdot
\end{pmatrix}$,
such that 
$\| \alpha \tilde{M} - M \| \leq  \epsilon.$
\end{definition}
If $\epsilon=0$, then $U$ may simply be called an $(\alpha, a)$-block-encoding of $M$. In the above definition, $\alpha \geq 1$ is called the subnormalization factor. $U$ uses $a$ extra qubits to implement $M$ up to total error $\epsilon$. 
The following result formalizes how to efficiently construct a block-encoding of a matrix stored in a quantum-accessible data structure. 

\begin{lemma}[Lemma 3.3.7 (2) in \cite{gilyen2019thesis}](Block-encoding matrices stored in QRAM)\label{lem:beQRAM}
Let $A \in \mathbb{C}^{2^n \times 2^n}$. If $A$ is stored in a quantum-accessible data structure, then there exists unitaries $U_R$ and $U_L$ that can be implemented with $\Ocal (\textup{poly}(n \log (1/\epsilon)))$ accesses to QRAM and additional gates, and $U_R^{\dagger} U_L$ is a $(\| A \|_F, n + 2, \epsilon)$-block-encoding of $A$. 
\end{lemma}

One can also block-encode matrices to which we have sparse oracle access. The next result from \cite{low2019hamiltonianOpt} describes this construction and the associated cost.
\begin{lemma}[Corollary 5 in \cite{low2019hamiltonianOpt}](Block-encoding via sparse access oracles)\label{lem:sparseBE}
Let $A \in \mathbb{C}^{2^n \times 2^n}$ and $\|A\| \le 1$. If we have sparse oracle access to $A$, then we can construct a $(\Theta(1), \textup{poly}(n), \epsilon)$-block-encoding of $A$ using $\Ocal (\sqrt{s} (s/\epsilon)^{o(1)})$ queries to the input oracles for $A$, and $\Ocal (\sqrt{s} (s/\epsilon)^{o(1)} \textup{poly}(n \log \frac{s}{\epsilon}))$ additional gates.
\end{lemma}

Lastly, it is straightforward to block-encode matrices to which we have an LCU decomposition. The next result from \cite{gilyen2019quantum} describes this construction and the associated cost.

\begin{lemma}[Specialization of Lemma 52 in \cite{gilyen2019quantum}](Block-encoding via a linear combination of unitaries)\label{lem:lcuBE}
Let $A = \sum_{i=0}^{K-1} \alpha_i U_i \in \Cmbb^{2^n \times 2^n}$ be a linear combination of $n$-qubit unitaries $U_i$ with coefficients $\alpha_i > 0$, and let $\alpha = \sum_{i=0}^{K-1} \alpha_i$. Assume that we have access to a state-preparation oracle $\mathrm{PREP}$ and a selection oracle $\mathrm{SEL}$ acting on a $\lceil \log K \rceil$-qubit ancilla register and the $n$-qubit system register:
\begin{align*}
     \mathrm{PREP} \ket{0}^{\otimes \lceil \log K \rceil} = \frac{1}{\sqrt{\alpha}} \sum_{i=0}^{K-1} \sqrt{\tilde{\alpha}_i} \ket{i},~\text{and}~
     \mathrm{SEL} = \sum_{i=0}^{K-1} \ket{i}\bra{i} \otimes U_i,
\end{align*}
where $\tilde{\alpha}_i$ approximate the true coefficients up to total error $\sum_{i=0}^{K-1} | \tilde{\alpha}_i - \alpha_i | \leq \epsilon$. Then
$$ U = \left( \mathrm{PREP}^{\dagger} \otimes I \right) \mathrm{SEL} \left( \mathrm{PREP} \otimes I \right) $$
is an $\left( \alpha, \lceil \log K \rceil, \epsilon \right)$-block-encoding of $A$, using a single use each of  $\mathrm{SEL}$, $\mathrm{PREP}$, and $\mathrm{PREP}^{\dagger}$, and $\Ocal(K)$ additional gates.
\end{lemma}

With access to a unitary block-encoding of the coefficient matrix $A$, we can perform the necessary matrix arithmetic operations that arise in solving a linear system of equations in time that depends only polylogarithmically on the dimension of the problem. In particular, the following result from \cite{dalzell2024shortcut} shows that one can use a block-encoding of the coefficient matrix $A$ to solve a linear system. 

\begin{theorem}[Adapted from Theorem 4 in \cite{dalzell2024shortcut}]
\label{thm:qlsa}
(Complexity of QLSA)
Assume that $\|b\|=1$, that $b$ is in the column space of $A$, and that
all nonzero singular values of $A$ lie in the interval
$[\frac{1}{\kappa},1]$. Let $U_A$ be an $(\alpha,a)$-block-encoding of
$A$, and let $U_b$ be a state-preparation unitary for $\ket{b}$. Let $x$
denote the unique vector of minimum norm for which $Ax=b$. Fix
$\epsilon>0$. Then there is a quantum algorithm to solve the QLSP to precision
$\epsilon$ using an expected $Q$ queries to $U_A$, $U_A^\dagger$, and
their controlled versions, as well as $2Q$ queries to $U_b$, $U_b^\dagger$,
and their controlled versions, where
\begin{equation*}
    Q \leq
    56.0\alpha\kappa
    +1.05\alpha\kappa
    \ln\left(\frac{\sqrt{1-\epsilon^2}}{\epsilon}\right)
    +2.78\ln(\alpha\kappa)^3
    +3.17.
\end{equation*}
The number of additional gates is bounded by $\widetilde{\Ocal}_{d,1/\epsilon}(\alpha \kappa)$.
\end{theorem}

\subsection{Quantum state tomography}\label{ss:tom}
To solve the LSP with a QLSA, we will require a classical estimate of $x$ from its amplitude encoding $\ket{x}$. This can be provided by a quantum state tomography algorithm (QTA), of which there are many. For the purposes of this paper, we want an algorithm that achieves low sample complexity (i.e., number of copies of the quantum state used in the process) as well as low quantum and classical computational complexity (i.e., gate count of the quantum circuits, and number of arithmetic operations for any classical pre- or post-processing). Taking advantage of the fact that we have access to a circuit that prepares the quantum state of interest, the current state of the art in terms of sample and computational complexities is given by the algorithms of van Apeldoorn et al.~\cite{van2023quantum}; other recent results with optimal sample complexity, e.g., \cite{guctua2020FastTomographyOptimal,pelecanos2025mixed}, rely on expensive classical operations that take time $\Ocal(d^3)$ for a $d$-dimensional quantum state, and would end up becoming a bottleneck for the IR scheme.

We are interested in obtaining a description of a quantum state with error at most $\epsilon$ measured with the Euclidean distance. \cite{van2023quantum} shows that $\Omega(d/\epsilon)$ copies of the quantum state are necessary for this task. We give a precise statement for the version of tomography that we use, achieving sample complexity linear in $d$ with or without QRAM. The gate complexity depends on the presence or absence of QRAM. Note that we only state the result for extracting the real part of the amplitudes; if we need the imaginary part as well (e.g., if the QLSA does not properly control the relative phase of the solution), we only need to apply the algorithm with extra phase gates that do not affect the asymptotic complexity.
\begin{theorem}[\cite{van2023quantum}]
\label{thm:euclidean_norm_tomo}
Let $\ket{\psi} = \sum_{j=0}^{d-1} v_j \ket{j}$ be a quantum state, $v \in \mathbb{C}^d$ the vector with elements $v_j$, and $U\ket{0} = \ket{\psi}$. There is a quantum algorithm that outputs $\tilde{v} \in \mathbb{R}^{d}$ such that $\|\Re(v) - \tilde{v}\|_2 \le \epsilon$ using $\Ocal(\frac{d}{\epsilon} \log d)$ applications of $U$. In the standard circuit model, the algorithm additionally uses $\Ocal(\frac{d^{2}}{\epsilon} \log \frac{d}{\epsilon})$ quantum gates. If QRAM is available, the gate complexity reduces to $\Ocal(\frac{d}{\epsilon} \log \frac{d}{\epsilon})$ gates and accesses to QRAM.
\end{theorem}
In fact, \cite{van2023quantum} uses indexed-SWAP gates, a QRAM-like operation that can be implemented with $\widetilde{\Ocal}(d)$ gates when acting on $d$ qubits, see \cite{van2023quantum} for details. Since we assume access to QRAM we also assume access to indexed-SWAP gates. (Alternatively, a QRAM of very large size would achieve the same result without indexed-SWAP gates.) We note that if $\ket{\psi}$ is a real vector (up to global phase), gate complexity $\widetilde{\Ocal}(\frac{d}{\epsilon^2})$ can be achieved in the standard gate model without QRAM, at the cost of using $\widetilde{\Ocal}(\frac{d}{\epsilon^2})$ applications of $U$ (i.e., a worsening of a factor $1/\epsilon$) \cite{gilyen24iterative}.

\subsection{A quantum oracle for the LSP}\label{LS:s:oracle}
Here, we define an oracle $O_{LS}$ that uses a QLSA, followed by QTA and rescaling of the solution vector, to obtain a classical solution to the LSP.  We will use this oracle extensively in the IR scheme presented in Section \ref{sec:IR}.

\begin{definition}[Quantum linear systems oracle]\label{LS:def:oracle}
Let $A \in \mathbb{R}^{d \times d}$ have singular values in $[1/(\alpha\kappa), 1/\alpha]$ for some $\alpha,
\kappa \ge 1$. Let $b$ in $\mathbb{R}^d$ be such that $||b|| = 1$ and $b$ is in the column space of $A$. Let $\xi \in (0,1)$. The quantum linear systems oracle
$O_{LS}(A, b, \xi)$ returns a classical vector
$\bar{x} \in \mathbb{R}^d$ satisfying
$$\left\| A \bar{x} - b \right\| \leq \xi.$$
\end{definition}

The next result bounds the query complexity of implementing $O_{LS}$.

\begin{theorem}[Complexity of $O_{LS}$]\label{LS:theo:oracle-complexity}
Let $A, b$ and $\xi$ be as in Definition \ref{LS:def:oracle}, let $U_A$ be a $(1,a)$-block encoding of $A/\alpha$ and let $U_b$ be a
state-preparation unitary for $\ket{b}$. Let $x$ denote the unique
vector of
minimum norm for which $Ax = b$. Then we can implement $O_{LS}(A, b, \xi)$ using
$$\widetilde{\Ocal}_{ d, \kappa_A, \tfrac{1}{\xi}}\left( \frac{\alpha \, d \, \kappa_A^2}{\xi} 
\right)$$
queries to $U_A$ and $U_b$, and quantum gates. In the standard circuit model without QRAM, the algorithm additionally uses $\widetilde{\Ocal}_{d, \kappa_A, \tfrac{1}{\xi} }\left(\frac{\alpha d^{2} \kappa_A^2}{\xi} \right)$ quantum gates.
\end{theorem}

\begin{proof}
The oracle proceeds in four steps: estimate the solution norm $\| x \|$,
prepare the normalized solution state $\ket{x}$, obtain its classical description via tomography,
and rescale.

First, the QLSA of Theorem~\ref{thm:qlsa} learns $\| x \|$ up to a constant
multiplicative factor as part of its execution (see \cite{dalzell2024shortcut}), and this estimate can be
refined sufficiently to a relative precision of $\Ocal(\xi)$ using amplitude estimation at
an additional cost of $\Ocal\left( \alpha \kappa_A \log(1/\xi)/\xi \right)$
queries \cite[Section~5.4]{dalzell2024shortcut}; this cost is dominated by the tomography
step below.

Next, we apply Theorem~\ref{thm:qlsa} to prepare a state $\ket{x^\ast}$ with
$\| \ket{x^\ast} - \ket{x} \| \leq \frac{\alpha \xi}{2 \| x \|}$; call $T_{LS}\left( \frac{\alpha \xi}{2 \| x \|} \right)$ the number of queries to prepare the state. Then we
apply the tomography algorithm of Theorem~\ref{thm:euclidean_norm_tomo} to
the circuit preparing $\ket{x^\ast}$, obtaining a classical unit vector
$\bar{v}$ with $\| \bar{v} - \ket{x^\ast} \| \leq \frac{\alpha \xi}{2 \| x \|}$; call $T_{TO}\left( \frac{\alpha \xi}{2 \| x \|} \right)$ the number of applications of the circuit that prepares $\ket{x^\ast}$. The oracle returns $\bar{x} \coloneqq \bar{\nu} \bar{v}$, where $\bar{\nu}$ is the norm estimate.

To verify the error guarantee of Definition~\ref{LS:def:oracle}, decompose
$$A \bar{x} - b
= \bar{\nu} A \left( \bar{v} - \ket{x} \right)
+ \left( \bar{\nu} - \| x \| \right) A \ket{x}.$$
For the first term, using $\| A \| \leq 1/\alpha$ and $\| \bar{v} - \ket{x} \| \leq \frac{\alpha \xi}{2 \| x \|}$ we obtain
$\|\bar{\nu} A \left( \bar{v} - \ket{x} \right)\| \le \bar{\nu} \cdot \frac{1}{\alpha} \cdot \frac{\alpha \xi}{2 \| x \|} =
(1 + \Ocal(\xi))\, \xi$. For the second term, $A \ket{x} = b / \| x \|$
exactly, so $$\|\left( \bar{\nu} - \| x \| \right) A \ket{x}\| = | \bar{\nu} - \| x \| | / \| x \| = \Ocal(\xi).$$
Hence $\| A \bar{x} - b \| \leq (1 + \Ocal(1))\, \xi$, and adjusting the
constants in the tolerances yields
$\| A \bar{x} - b \| \leq \xi$.

For the complexity, the total number of applications of $U_A$, $U_b$, their
inverses and controlled versions is the number of copies of $\ket{x^\ast}$ consumed by tomography times the cost for preparing $\ket{x^\ast}$:
\begin{align*}
T_{TO}\left( \tfrac{\alpha \xi}{2 \| x \|} \right) \cdot
T_{LS}\left( \tfrac{\alpha \xi}{2 \| x \|} \right)
&= \Ocal\left( \tfrac{d \, \| x \|}{\alpha \xi}
\log d \right) \cdot
\Ocal\left( \alpha \kappa_A
\log\left( \tfrac{\| x \|}{\alpha \xi} \right) \right) \\
&= \Ocal\left( \tfrac{\alpha \, d \, \kappa_A^2}{\xi} \cdot
\operatorname{polylog}\left( d, \kappa_A, \tfrac{1}{\xi} \right) \right),
\end{align*}
where we used the fact that $\|x\| \le \|A^{-1}\| \|b\| \le \alpha \kappa_A$. The number of gates follows from Theorems~\ref{thm:qlsa} and \ref{thm:euclidean_norm_tomo}. 
\end{proof}

Combining Lemmas~\ref{lem:beQRAM}, \ref{lem:sparseBE}, \ref{lem:lcuBE} with Theorem~\ref{LS:theo:oracle-complexity}, we obtain the following.
\begin{corollary}[Complexity of $O_{LS}$ in the QRAM input model]
\label{cor:ols-qram}
Let $A \in \R{d \times d}$ have nonzero singular values in $[1/\kappa_A, 1]$. Let $b \in \R{d}$ be such that $\|b\|=1$ and $b$ is in the column space of $A$. In the QRAM input model we can implement $O_{LS}(A, b, \xi)$ with
$$
    \widetilde{\Ocal}_{d,\kappa_A,\frac{1}{\xi}}
    \left(
        \frac{\| A \|_F d \kappa_A^2}{\xi}
    \right)
$$
accesses to QRAM and additional gates.
\end{corollary}
\begin{corollary}[Complexity of $O_{LS}$ in the sparse access input model]
\label{cor:ols-sparse}
Let $A \in \R{d \times d}$ have nonzero singular values in $[1/\kappa_A, 1]$. Let $b \in \R{d}$ be such that $\|b\|=1$ and $b$ is in the column space of $A$. In the sparse access input model we can implement $O_{LS}(A, b, \xi)$ with
$$
    \widetilde{\Ocal}_{d,\kappa_A,\frac{1}{\xi}}
    \left(
        \frac{\sqrt{s} d \kappa_A^2 }{\xi}(s \kappa_A/\xi)^{o(1)}
    \right)
$$
accesses to the sparse oracles of Section~\ref{s:sparse} and 
$
    \widetilde{\Ocal}_{d,\kappa_A,\frac{1}{\xi}}
    \left(
        \frac{\sqrt{s} d^2 \kappa_A^2}{\xi}
    \right)
$
additional gates.
\end{corollary}
\begin{corollary}[Complexity of $O_{LS}$ in the LCU input model]
\label{cor:ols-lcu}
Let $A = \sum_{i=0}^{K-1} \alpha_i U_i \in \R{d \times d}$ have nonzero singular values in $[1/\kappa_A, 1]$, and let $\alpha = \sum_{i=0}^{K-1}|\alpha_i|$. Let $b \in \R{d}$ be such that $\|b\|=1$ and $b$ is in the column space of $A$. In the LCU input model we can implement $O_{LS}(A, b, \xi)$ with
$$
    \widetilde{\Ocal}_{d,\kappa_A,\frac{1}{\xi}}
    \left(
        \frac{\alpha d \kappa_A^2}{\xi}
    \right)
$$
accesses to the LCU oracles of Lemma~\ref{lem:lcuBE} and 
$\widetilde{\Ocal}_{d,\kappa_A,\frac{1}{\xi}}
    \left(
        \frac{\alpha d^2 \kappa_A^2}{\xi}
    \right)$
additional gates.
\end{corollary}
\begin{proof}
This proof is slightly more involved than the previous ones because we need to account for the error in the construction of the block-encoding
via Lemma~\ref{lem:lcuBE}. We choose a coefficient error $\delta \leq \frac{\xi}{8 \kappa_A}$ in the unitary $\textup{PREP}$, so that the implemented unitary is an exact $(1, \lceil \log K \rceil)$-block-encoding of a perturbed matrix $A_\delta$, satisfying $\| A_\delta - A/\alpha \| \leq \delta / \alpha \leq \frac{\xi}{8\alpha\kappa_A}$. We apply Theorem~\ref{LS:theo:oracle-complexity} to $A_\delta$ with precision $\xi/2$. This returns a vector $\bar{u}$ with bounded norm $\|\bar{u}\| = \Ocal(\alpha \kappa_A)$. We have:
$$ \| \frac{A}{\alpha} \bar{u} - b \| \leq \| A_\delta \bar{u} - b \| + \| \frac{A}{\alpha} - A_\delta \| \| \bar{u} \| \leq \frac{\xi}{2} + \frac{\xi}{8\alpha\kappa_A} \Ocal(\alpha \kappa_A) \leq \xi.$$
\end{proof}

\section{Iterative Refinement for the LSP using quantum subroutines}
\label{sec:IR}

Iterative refinement is a classical computing technique used to obtain
accurate solutions to systems of linear equations. This section summarizes
the core idea at a high level before providing an analysis of our
algorithm.

\subsection{The algorithm}

Our IR method for the LSP is presented in detail in
Algorithm \ref{LS:alg:iterative refinement for LS}. In every iteration,
there is a call to the quantum linear systems oracle $O_{LS}$, which we
describe in detail in Section \ref{LS:s:oracle}. There are two precision
parameters: $\xi$, the fixed precision we use for each of our oracle calls,
and $\epsilon$, the desired precision of the final solution. We assume the problem data $(A, b)$ satisfies the hypotheses of Definition~\ref{LS:def:oracle} (after normalizing $b$), and that $A$ is
invertible; the singular case is handled by replacing $A^{-1}$ with the
pseudoinverse $A^{+}$ throughout.

The target is to find a classical solution $x$ such that
\[
    \lVert Ax-b\rVert \leq \epsilon.
\]

We classically scale the linear system by setting
\[
    \tilde{A}\gets \frac{A}{\alpha},
    \qquad
    \tilde{b}\gets \frac{b}{\lVert b\rVert},
\]
and use our IR algorithm to find a solution $\tilde{x}$ of the scaled
system. Then, the solution of the original problem can be recovered as
\[
    x=\frac{\lVert b\rVert}{\alpha}\tilde{x}.
\]
Here, $\alpha$ is the subnormalization factor of the block-encoding of
$A$. Note that in all of the input models considered here, $\alpha$ is explicitly known or can be estimated with classical operations whose cost is smaller than the cost of $O_{LS}$, hence we can ignore this aspect in the analysis. The classical scaling $x=\lVert b\rVert \tilde{x}/\alpha$ leads to better complexity for the quantum
oracle. To achieve
$\lVert Ax-b\rVert\leq\epsilon$, we need
\[
    \lVert\tilde{A}\tilde{x}-\tilde{b}\rVert
    \leq \frac{\epsilon}{\lVert b\rVert};
\]
if we aim to achieve
$\lVert x-A^{-1}b\rVert\leq\epsilon'$, then we need  
\[
    \lVert\tilde{x}-\tilde{A}^{-1}\tilde{b}\rVert
    \leq \frac{\epsilon'\alpha}{\lVert b\rVert}.
\]
As $ \lVert\tilde{x}-\tilde{A}^{-1}\tilde{b}\rVert \leq \alpha \kappa_A \lVert\tilde{A}\tilde{x}-\tilde{b}\rVert$, then IR requires

$$\lVert\tilde{A}\tilde{x}-\tilde{b}\rVert \leq \frac{\epsilon'}{\kappa_A\lVert b\rVert}$$
This allows us to set the target precision of the IR scheme.

The algorithm starts from an initial point
$\tilde{x}^{(0)}\in\mathbb{R}^{d}$ that can be chosen arbitrarily, but it is reasonable to set to it be the all-zero vector. In each iteration, we first compute
the residual and its norm classically as
\[
    r^{(k)}
    =
    \tilde{b}-\tilde{A}\tilde{x}^{(k)}.
\]
We then prepare the normalized residual state
\[
    |r^{(k)}\rangle
    =
    \frac{r^{(k)}}{\lVert r^{(k)}\rVert}
\]
and call $O_{LS}$ to approximately solve the normalized refining system
\begin{equation}
    \tilde{A}u^{(k)}
    \approx
    \frac{r^{(k)}}{\lVert r^{(k)}\rVert}.
    \label{LS:e:corrSystem}
    \tag{Ref-Sys}
\end{equation}
An approximate solution $\tilde{u}^{(k)}$ of \eqref{LS:e:corrSystem} therefore provides a
correction direction for the current residual. We update the current solution according to
\[
    \tilde{x}^{(k+1)}
    \gets
   \tilde{x}^{(k)}
    +
    \lVert r^{(k)}\rVert \tilde{u}^{(k)}.
\]

This section shows that we can use a fixed level of precision (e.g., $\xi=10^{-1}$) for every call
to $O_{LS}$ in Algorithm \ref{LS:alg:iterative refinement for LS}.
In particular, while the desired final accuracy $\epsilon$ may be
arbitrarily small, the quantum linear systems oracle is only required
to solve each refining system to the fixed precision $\xi$. Once
    $\lVert r^{(k)}\rVert
    \leq
    \frac{\epsilon}{\lVert b\rVert},$
the algorithm terminates and reports the corresponding solution.

\begin{algorithm}
\begin{algorithmic}
\STATE \textbf{Input}: Coefficient matrix
$A\in\mathbb{R}^{d\times d}$, right-hand side vector
$b\in\mathbb{R}^{d}$, error tolerances $0<\epsilon\ll\xi<1$.

\STATE \textbf{Output}: A classical vector $x\in\mathbb{R}^{d}$
satisfying $\lVert Ax-b\rVert\leq\epsilon$.\\

\STATE Normalize the system
$\tilde{b}\gets b/\lVert b\rVert$ and
$\tilde{A}\gets A/\alpha$.\\

\STATE Choose $\tilde{x}^{(0)}\gets 0$ and $k\gets 0$.\\

\STATE Compute
$r^{(0)}\gets\tilde{b}-\tilde{A}\tilde{x}^{(0)}$
and 
$\lVert r^{(0)}\rVert$.\\

\textbf{while} 
$\lVert r^{(k)}\rVert>\frac{\epsilon}{\lVert b\rVert}$
\begin{enumerate}

\item Classically compute the residual
\[
    r^{(k)}
    \gets
    \tilde{b}-\tilde{A}\tilde{x}^{(k)}
\]
and its norm $\lVert r^{(k)}\rVert.$

\item Prepare the normalized residual state
\[
    |r^{(k)}\rangle
    =
    \frac{r^{(k)}}{\lVert r^{(k)}\rVert}.
\]

\item Use $O_{LS}(\tilde{A}, r^{(k)}/\lVert r^{(k)}\rVert, \xi)$ to obtain a classical vector
$\tilde{u}^{(k)}$ satisfying
\[
    \left\lVert
    \tilde{A}\tilde{u}^{(k)}
    -
    \frac{r^{(k)}}{\lVert r^{(k)}\rVert}
    \right\rVert
    \leq \xi.
\]
\label{LS:step:solve}

\item Update
\[
    \tilde{x}^{(k+1)}
    \gets
    \tilde{x}^{(k)}
    +
    \lVert r^{(k)}\rVert
    \tilde{u}^{(k)}.
\]

\item $k\gets k+1$.

\end{enumerate}

\textbf{end}\\

\STATE \textbf{Return}
$x=\frac{\lVert b\rVert}{\alpha}\tilde{x}^{(k)}$.

\caption{Iterative Refinement for the Linear System Problem}
\label{LS:alg:iterative refinement for LS}
\end{algorithmic}
\end{algorithm}

As shown in Algorithm \ref{LS:alg:iterative refinement for LS}, the
residual vector and its norm are computed classically at every iteration.
The normalized residual is then encoded into the quantum state
$|r^{(k)}\rangle$ and provided as the right-hand side to the quantum
linear systems oracle. The oracle solves the refining system only to the
fixed precision $\xi$, after which tomography is used to obtain the
classical correction vector $\tilde{u}^{(k)}$. The correction is then
rescaled by the residual norm and added classically to the current
iterate. Thus, high precision is achieved through the classical
iterative-refinement procedure, while each quantum linear-system solve
is performed only at fixed precision.

\subsection{Convergence}

In the following result, we establish that the sequence of solutions generated by Algorithm \ref{LS:alg:iterative refinement for LS} satisfies the normalized linear system $A x = \tilde{b}$ with increasing accuracy.
\begin{theorem}\label{LS:theo; II-LS}
   Let $\tilde{x}^{(k)}$ and $r^{(k)}$, $k = 0, 1, \dots$ be the sequence of solutions and residual produced by Algorithm \ref{LS:alg:iterative refinement for LS}. Then for all $k \geq 0$,
$$\left\| \tilde{b} - \tilde{A} \tilde{x}^{(k)} \right\| \leq \xi^{k}.$$

\end{theorem}
\begin{proof}
    By induction. First, observe that $\|r^{(0)}\|=\|\tilde{b}\|=1$ trivially holds at $k = 0$. Now assume that $\|r^{(k)}\|\leq  \xi^{k}$ is true for any $k\geq 0$ under the induction hypothesis, and check the residual at iteration $k+1$. 
    
    The classical estimate $\tilde{u}^{(k)}$ satisfies
    $$\left\|\frac{ r^{(k)}}{\|r^{(k)}\|} - \tilde{A} \tilde{u}^{(k)}\right\|\leq  \xi.$$
    Therefore, it follows that the updated solution in each iteration satisfies
   \begin{align*}
        \left\|r^{(k+1)}\right\| =  \left\|\tilde{b}-\tilde{A}\tilde{x}^{(k+1)} \right\|
        &= \left\|\tilde{b}-\tilde{A}\left(\tilde{x}^{(k)}+
        \lVert r^{(k)}\rVert\tilde{u}^{(k)}\right)\right\|
        =\left\|\tilde{b} - \tilde{A} \tilde{x}^{(k)} - \lVert r^{(k)}\rVert
        \tilde{A} \tilde{u}^{(k)} \right\|  \\
        &=\|r^{(k)}\| \left\|\frac{ r^{(k)}}{\|r^{(k)}\|} - \tilde{A} \tilde{u}^{(k)} \right\| 
        \leq \|r^{(k)}\| \xi\leq \xi^{k+1}.
   \end{align*}
\end{proof}

\begin{corollary}\label{LS:corr: iteration complexity}
Let $\xi > 0$ be fixed. Let $0 < \epsilon \ll \xi < 1$. Then, Algorithm~\ref{LS:alg:iterative refinement for LS} terminates in at
most
$$\left\lceil \frac{\log\left( \| b \| / \epsilon \right)}{\log\left( 1/\xi \right)} \right\rceil
= \Ocal\left( \log\left( \frac{\| b \|}{\epsilon} \right) \right)$$
iterations.
\end{corollary}

\begin{proof}
By Theorem~\ref{LS:theo; II-LS}, $\| r^{(k)} \| \leq \xi^{k}$, and
Algorithm~\ref{LS:alg:iterative refinement for LS} terminates once $\| r^{(k)} \| \leq \epsilon / \| b \|$.
This is guaranteed as soon as $\xi^{k} \leq \epsilon / \| b \|$, i.e., for
$k \geq \log\left( \| b \| / \epsilon \right) / \log\left( 1/\xi \right)$.
\end{proof}
The iteration bound provided in Corollary \ref{LS:corr: iteration complexity} is independent of the chosen input model.  When the forward-error guarantee $\left\| x - A^{-1} b \right\| \leq \epsilon$
is sought, Algorithm~\ref{LS:alg:iterative refinement for LS}
is invoked with tolerance $\epsilon / \kappa_A$ in place of $\epsilon$: since
$\left\| x - A^{-1} b \right\| \leq \left\| A^{-1} \right\| \left\| A x - b \right\|
\leq \kappa_A \left\| A x - b \right\|$, a residual of $\epsilon / \kappa_A$
suffices, and the iteration bound of Corollary~\ref{LS:corr: iteration complexity} becomes
$\Ocal\left( \log\left( \tfrac{\| b \| \kappa_A}{\epsilon} \right) \right)$.

\subsection{Computational complexity}\label{sec:complexity}

We now analyze the computational complexity of
Algorithm \ref{LS:alg:iterative refinement for LS}. At each iteration,
the residual
\[
    r^{(k)}
    =
    \tilde{b}-\tilde{A}\tilde{x}^{(k)}
\]
and its norm
are computed classically. The normalized residual
\[
    \frac{r^{(k)}}{\lVert r^{(k)}\rVert}
\]
is used as the right-hand side of
the quantum linear-system oracle $O_{LS}$, which requires us to prepare a quantum state with amplitudes given by $r^{(k)}/\lVert r^{(k)}\rVert$. Note that the computation of the residual does not contribute to
the query complexity of the block-encoding oracle.

More precisely, one iteration of Algorithm
\ref{LS:alg:iterative refinement for LS} consists of the following
operations. First, the matrix-vector product
$\tilde{A}\tilde{x}^{(k)}$ is computed classically and used to form
$r^{(k)}$. This takes $\mathcal{O}(\operatorname{nnz}(A)) = \mathcal{O}(ds)$ arithmetic operations. The
normalization factor $\lVert r^{(k)}\rVert$ is also computed
classically.

The normalized residual vector is subsequently encoded into the quantum state
\[
    |r^{(k)}\rangle
    =
    \frac{1}{\lVert r^{(k)}\rVert}
    \sum_{j=1}^{d}r_j^{(k)}|j\rangle.
\]
We denote the cost of preparing this state from its classical
description by $T_{\mathrm{SP}}$. In the QRAM input model, once the
classical data structure associated with $r^{(k)}$ has been constructed,
the corresponding amplitude-encoded state can be prepared with
polylogarithmic QRAM query complexity in the dimension
\cite{kerenidis2017quantum}.

After preparing $|r^{(k)}\rangle$, the algorithm calls oracle $O_{LS}(\tilde{A},\allowbreak r^{(k)}/\lVert r^{(k)}\rVert, \xi)$.
The resulting quantum state is converted into a classical correction vector using tomography. Since $\xi$ is fixed throughout the iterative-refinement procedure, the quantum cost of each refining solve is independent of the final target precision $\epsilon$. The complexity of solving one refining system can be derived directly from Theorem~\ref{LS:theo:oracle-complexity}, taking $\xi$ fixed.

\subsection{A difference between classical and quantum IR}
In classical IR sche\-mes, one has the choice of the norm used to compute the error incurred by the linear system solution, i.e., the norm used to assess the size of the residual. The choice of the $\ell^\infty$-norm is common \cite{demmel2006error}, and the convergence proof carries through without any impediment. In the quantum scheme described in Algorithm~\ref{LS:alg:iterative refinement for LS}, however, the right-hand side vector $r^{(k)}$ is always normalized in the $\ell^2$-norm because it is encoded as a quantum state. As a consequence, we cannot measure errors in a different norm. We show where the proof of Theorem~\ref{LS:theo; II-LS} fails if we wish to employ the $\ell^\infty$-norm. The classical estimate $\tilde{u}^{(k)}$ satisfies $\left\|\frac{ r^{(k)}}{\|r^{(k)}\|_2} - \tilde{A} \tilde{u}^{(k)}\right\|_\infty \leq  \xi.$ The induction step reads
\begin{align*}
        \left\|r^{(k+1)}\right\|_\infty =  \left\|\tilde{b}-\tilde{A}\tilde{x}^{(k+1)} \right\|_\infty
        &= \left\|\tilde{b}-\tilde{A}\left(\tilde{x}^{(k)}+
        \lVert r^{(k)}\rVert_2\tilde{u}^{(k)}\right)\right\|_\infty \\
        &=\left\|\tilde{b} - \tilde{A} \tilde{x}^{(k)} - \lVert r^{(k)}\rVert_2
        \tilde{A} \tilde{u}^{(k)} \right\|_\infty  \\
        &=\|r^{(k)}\|_2 \left\|\frac{ r^{(k)}}{\|r^{(k)}\|_2} - \tilde{A} \tilde{u}^{(k)} \right\|_\infty
        \leq \|r^{(k)}\|_2 \xi\leq \sqrt{d} \xi^{k+1},
\end{align*}
whereas our goal was to obtain the tighter bound $\left\|r^{(k+1)}\right\|_\infty \le \xi^{k+1}$. This observation is particularly important in light of the fact that if we were only looking for an $\ell^\infty$-norm error bound for the solution of the linear system, the tomography step would become computationally much cheaper: \cite{van2023quantum} shows that $\widetilde{\Ocal}(1/\xi^2)$ applications of the circuit to prepare the quantum state are sufficient to obtain a classical estimate with $\ell^\infty$-norm error at most $\xi$. Note that this is independent of the dimension $d$. Unfortunately, the discussion above shows that an $\ell^\infty$-norm error bound at each iteration does not suffice for our purposes, due to the $\ell^2$-norm normalization of the residual.

\section{Complexity analysis in the three input models}
\label{sec:complexity_analysis_input_models}
We can now characterize the total complexity of algorithm \ref{LS:alg:iterative refinement for LS} under the QRAM, sparse access, and LCU input models, and compare it to results in the open literature.

\subsection{Detailed complexity statements}
In this section, as usual we assume that the singular values of $A$ lie in $[1/\kappa_A, 1]$, and we take into account the normalizations applied by Algorithm~\ref{LS:alg:iterative refinement for LS}. 

\begin{theorem}\label{LS:theo: complexity QRAM}
Let $A \in \R{d \times d}$ be $s$-sparse, and let $b \in \R{d}$. Let $\epsilon > 0$. In the QRAM input model, Algorithm \ref{LS:alg:iterative refinement for LS} classically outputs a vector $x \in \R{d}$ satisfying $\left\| x - A^{-1} b \right\| \leq \epsilon$ using at most
$$ \tilde{\Ocal}_{d, \kappa_A, \frac{ \| b \|}{\epsilon}} \left(  \|A\|_F d \kappa^2_A  \right)$$
accesses to QRAM and additional quantum gates, and 
$$ \tilde{\Ocal}_{\kappa_A, \frac{ \| b \|}{\epsilon}} \left(  d s  \right)$$
classical arithmetic operations. 
\end{theorem}

Proof of this theorem follows by combining  Corollary~\ref{cor:ols-qram} and Theorem~\ref{LS:theo; II-LS}. Note that because $\|A\| \le 1$, then $\|A\|_F \le \sqrt{d}$.

\begin{theorem}\label{LS:theo: complexity sparse}
Let $A \in \R{d \times d}$ be $s$-sparse, and let $b \in \R{d}$. Let $\epsilon > 0$. In the sparse access input model, Algorithm \ref{LS:alg:iterative refinement for LS} classically outputs a vector $x \in \R{d}$ satisfying $\left\| x - A^{-1} b \right\| \leq \epsilon$ using at most
$$ \tilde{\Ocal}_{d, \kappa_A, \frac{ \| b \|}{\epsilon}} \left( \sqrt{s} d \kappa^2_A  (s\kappa_A)^{o(1)}\right)$$
accesses to sparse access oracles,
$$ \tilde{\Ocal}_{d, \kappa_A, \frac{ \| b \|}{\epsilon}} \left(\sqrt{s} d^2 \kappa_A^2 \right)$$
additional quantum gates, and 
$$ \tilde{\Ocal}_{\kappa_A, \frac{ \| b \|}{\epsilon}} \left(  d s  \right)$$
classical arithmetic operations. 
\end{theorem}

Proof of this theorem follows by combining  Corollary~\ref{cor:ols-sparse} and Theorem~\ref{LS:theo; II-LS}.

\begin{theorem}\label{LS:theo: complexity LCU}
Let $A \in \R{d \times d}$, and let $b \in \R{d}$. Let $\epsilon > 0$. In the LCU input model, Algorithm \ref{LS:alg:iterative refinement for LS} classically outputs a vector $x \in \R{d}$ satisfying $\left\| x - A^{-1} b \right\| \leq \epsilon$ using at most
$$ \tilde{\Ocal}_{d, \kappa_A, \frac{ \| b \|}{\epsilon}} \left(  \alpha d \kappa^2_A  \right)$$
accesses to LCU oracles,
$$ \tilde{\Ocal}_{d, \kappa_A, \frac{ \| b \|}{\epsilon}} \left(\alpha d^2 \kappa_A^2 \right)$$
additional quantum gates, and 
$$ \tilde{\Ocal}_{\kappa_A, \frac{ \| b \|}{\epsilon}} \left(K C_U\right)$$
classical arithmetic operations, where $C_U$ is an upper bound on the number of arithmetic operations necessary to compute each matrix-vector product with the matrices $U_i$ of the LCU decomposition.
\end{theorem}
The proof follows by combining  Corollary~\ref{cor:ols-lcu} and Theorem~\ref{LS:theo; II-LS}; note that we assume that the coefficients $\alpha_i$ are known classically and we have some type of classical description of the unitaries $U_i$ as well.

\subsection{Comparison to existing linear system solvers}
Table \ref{tab:LS2} compares the running time of our algorithm to other classical and quantum algorithms in the literature when applied to classically solving linear systems of equations. The proposed IR scheme is exponentially faster than the existing quantum approach, regardless of whether or not we have access to QRAM. 

The classical CG method overall has the fastest asymptotic running time: the quantum IR scheme described in this paper does not achieve quantum advantage. However, it represents a step toward potential quantum advantage, due to its polylogarithmic dependence on the inverse of the precision to which we solve the linear system, compared to previous quantum approaches that incur polynomial cost in these quantities to obtain a classical description of the solution. Further improvements may be able to decrease the dependence on $\kappa$ and show an asymptotic advantage compared to classical CG; in particular, as mentioned in Section~\ref{ss: CT}, \cite{gilyen24iterative} relies on a preliminary version of this paper (\cite{mohammadisiaroudi2023exponentially}) to improve the $\kappa$ dependence.

\begin{table*}[t]
\centering
\caption{Complexity of algorithms for classically solving the LSP to
$\epsilon$-precision. For the proposed IR-based method, the dependence on
the final precision $\epsilon$ is contained only in polylogarithmic factors.
For quantum algorithms, the quantum complexity reports oracle/QRAM queries
and, when relevant, additional gate complexity.  As before, we use the assumption that $\|A\| \le 1$ so $\|A\|_F \le \sqrt{d}$.}
\label{tab:LS2}
\renewcommand{\arraystretch}{1.18}
\resizebox{\textwidth}{!}{
\begin{tabular}{p{5cm}p{3cm}p{4.0cm}p{7.0cm}}
\toprule
\textbf{Algorithm} &
\textbf{Input Model} &
\textbf{Classical Complexity} &
\textbf{Quantum Complexity} \\
\midrule

Factorization methods (dense)
&
Classical
&
$\Ocal(d^3)$
&
-- 
\\

Fast matrix multiplication ($\omega<2.372864$)
\cite{le2014powers}
&
Classical
&
$\Ocal(d^\omega)$
&
--
\\

Fast sparse solver
\cite{peng2021solving}
&
Classical sparse
&
$\widetilde{\Ocal}_{d,\kappa_A,1/\epsilon}
\!\left(d^{2.331645}s\right)$
&
--
\\

Conjugate Gradient ($A\succ0$)
\cite{hestenes1952methods}
&
Classical sparse
&
$\widetilde{\Ocal}_{d,\kappa_A,1/\epsilon}
\!\left(ds\sqrt{\kappa_A}\right)$
&
--
\\

Chebyshev polynomial
\cite{saad2003iterative}
&
Classical sparse
&
$\widetilde{\Ocal}_{d,\kappa_A,1/\epsilon}
\!\left(ds\kappa_A\right)$
&
--
\\
\midrule

QTA+QLSA (\cite{dalzell2024shortcut})
&
QRAM
&
&
$\widetilde{\Ocal}_{d,\kappa_A,1/\epsilon}
\!\left(\|A\|_F d\kappa_A^21/\epsilon\right)$ QRAM accesses\\

QTA+QLSA (\cite{dalzell2024shortcut})
&
Sparse access
&
&
$\widetilde{\Ocal}_{d,\kappa_A,1/\epsilon}
\!\left(\sqrt{s}d\kappa_A^2 1/\epsilon\right)$ oracle queries,
plus
$\widetilde{\Ocal}_{d,\kappa_A,1/\epsilon}
\!\left(\sqrt{s}d^2\kappa_A^2 1/\epsilon\right)$ gates
\\

QTA+QLSA (\cite{dalzell2024shortcut})
&
LCU
&
&
$\widetilde{\Ocal}_{d,\kappa_A,1/\epsilon}
\!\left(\alpha d\kappa_A^2 1/\epsilon\right)$ input-oracle accesses,
plus
$\widetilde{\Ocal}_{d,\kappa_A,1/\epsilon}
\!\left(\alpha d^2 \kappa_A^2 1/\epsilon\right)$ gates
\\

\midrule

IR--QTA+QLSA (this work)
&
QRAM
&
$\widetilde{\Ocal}_{\kappa_A,\|b\|/\epsilon}
\!\left(ds\right)$
&
$\widetilde{\Ocal}_{d,\kappa_A,\|b\|/\epsilon}
\!\left(\|A\|_F d\kappa_A^2\right)$ QRAM accesses\\

IR--QTA+QLSA (this work)
&
Sparse access
&
$\widetilde{\Ocal}_{\kappa_A,\|b\|/\epsilon}
\!\left(ds\right)$
&
$\widetilde{\Ocal}_{d,\kappa_A,\|b\|/\epsilon}
\!\left(\sqrt{s}d\kappa_A^2\right)$ oracle queries,
plus
$\widetilde{\Ocal}_{d,\kappa_A,\|b\|/\epsilon}
\!\left(\sqrt{s}d^2\kappa_A^2\right)$ gates
\\

IR--QTA+QLSA (this work)
&
LCU
&
$\widetilde{\Ocal}_{\kappa_A,\|b\|/\epsilon}
\!\left(KC_U\right)$
&
$\widetilde{\Ocal}_{d,\kappa_A,\|b\|/\epsilon}
\!\left(\alpha d\kappa_A^2\right)$ input-oracle accesses,
plus
$\widetilde{\Ocal}_{d,\kappa_A,\|b\|/\epsilon}
\!\left(\alpha d^2 \kappa_A^2 \right)$ gates
\\

\bottomrule
\end{tabular}
}
\end{table*}
\section{Numerical Experiments}\label{sec:num}


In this section, we empirically validate the iterative refinement scheme
presented in Algorithm~\ref{LS:alg:iterative refinement for LS} by solving randomly generated linear
systems on both quantum computer simulators and quantum hardware.  Our implementation uses
the HHL algorithm \cite{harrow2009quantum} as the underlying QLSA, using textbook quantum phase estimation (QPE) to approximate eigenvalues of~$A$ to $q$ bits of precision and quantum state tomography to extract a classical solution vector at each iteration.  All code is publicly available on \href{https://github.com/QCOL-LU/QLSAs}{GitHub} \footnote{https://github.com/QCOL-LU/QLSAs}.

We generate test instances $Ax = b$ where $A \in \mathbb{R}^{d \times d}$ is a random Hermitian matrix with prescribed condition number~$\kappa_A$ and sparsity~$s$, and $b$ is a random right-hand side vector normalized to unit norm.  The HHL circuit is constructed using Qiskit and transpiled at optimization level~3.  At each IR iteration, we run $500$ successful shots conditioned (post-selected) upon a successful mid-circuit ancilla measurement; this eliminates the need for amplitude amplification, while still serving as a proof-of-concept evaluation of our framework.

\subsection{Convergence in simulation and quantum hardware}
 
To demonstrate that Algorithm~\ref{LS:alg:iterative refinement for LS} works in practice while also being robust to hardware noise, we execute the full iterative refinement loop both in noiseless simulation and on quantum hardware. For simulation, we use the \texttt{AerSimulator} noiseless statevector backend; for hardware experiments, we execute on the IBM \texttt{ibm\_pittsburgh} superconducting quantum computer (a 156-qubit Heron~r3 processor). The test instance is an $8 \times 8$ random Hermitian system ($d = 8$) with condition number $\kappa_A = 5$ and sparsity~$4$. To probe the sensitivity of the scheme to the precision of the quantum oracle, we sweep the size of the QPE register over $q \in \{3,4,5\}$ qubits, using a total of $\log_2(d) + q + 1 \in \{7,8,9\}$ qubits for the full QLSA; after transpilation to the device connectivity these circuits contain $1327$, $1820$, and $2353$ two-qubit gates, respectively.  All hardware runs employ dynamical decoupling together with Pauli twirling of the two-qubit gates for error mitigation.

\begin{figure}[htbp]
    \centering
    \includegraphics[width=0.99\textwidth]{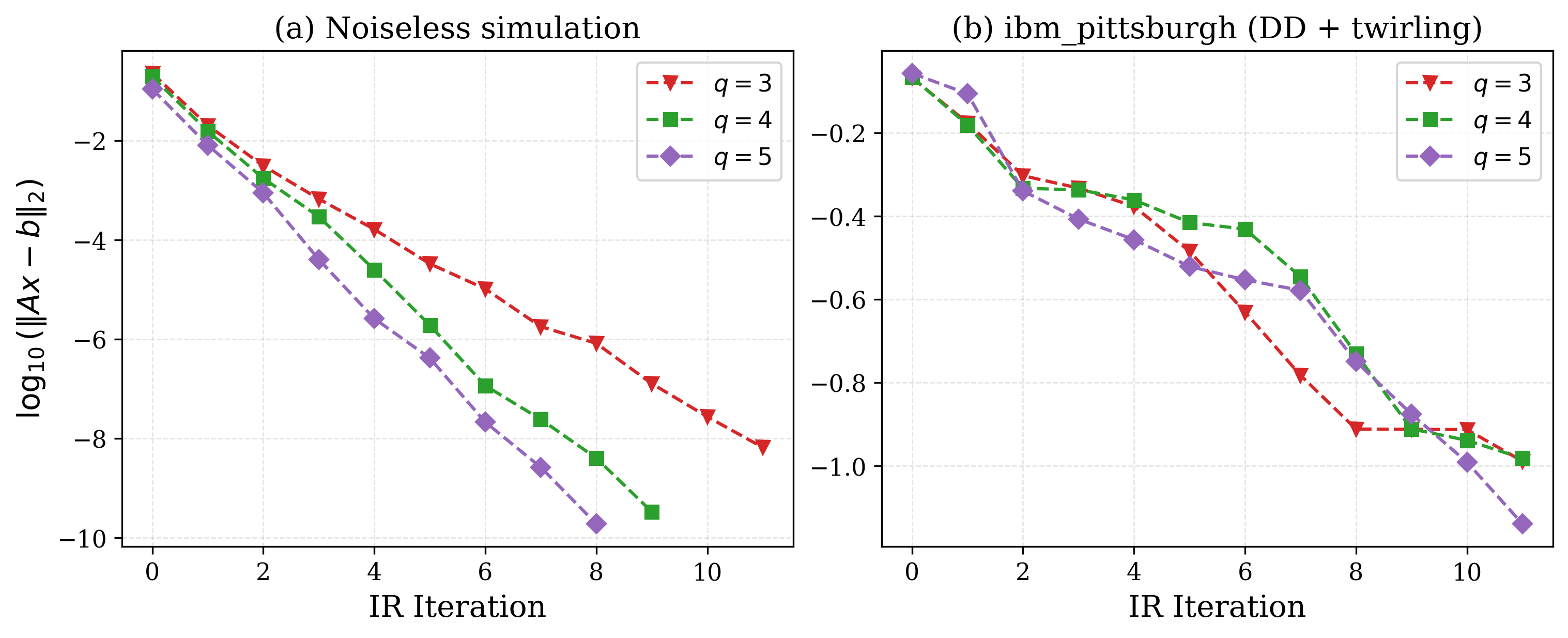}
    \caption{Residual norm convergence for an $8$-variable linear system with
    $q \in \{3,4,5\}$ QPE qubits over $12$ iterative refinement iterations.
    (a)~Noise-free simulation, in which each additional QPE qubit accelerates
    convergence. (b)~The same three configurations executed on the
    \texttt{IBM Pittsburgh} quantum computer.}
    \label{fig:hardware_convergence}
\end{figure}

Figure~\ref{fig:hardware_convergence} shows the residual norm
$\log_{10}\|Ax^{(k)} - b\|_2$ as a function of the IR iteration index $k$. In
noise-free simulation (panel~(a)) the convergence rate is ordered by the precision
of the quantum oracle as expected. Initializing the algorithm with a $10^{-9}$ residual norm convergence criteria and a $12$-iteration budget, the algorithm attains the desired residual norm in $9$ iterations for $q = 5$ and $10$ iterations for $q = 4$,
while $q = 3$ reaches only $6.6 \times 10^{-9}$ within the iteration budget.
Each additional QPE qubit decreases the eigenvalue quantization error and therefore
yields a more accurate search direction at every iteration.

On hardware (panel~(b)) all three configurations converge monotonically, reducing
the residual by roughly an order of magnitude from $8.6 \times 10^{-1}$ to
$1.03 \times 10^{-1}$ ($q = 3$), $1.04 \times 10^{-1}$ ($q = 4$), and
$7.3 \times 10^{-2}$ ($q = 5$).  However, the three curves are statistically
indistinguishable from one another. This is due to the device error rate: at a median two-qubit gate error of $1.5 \times 10^{-3}$, the transpiled circuits accumulate between two and four expected two-qubit errors per execution, so the fidelity of the returned residual at each IR iteration is bottlenecked by decoherence rather than by the resolution of the QPE register. Increasing $q$ improves the eigenvalue estimate but simultaneously deepens the circuit, and at current error rates the decoherence of the deeper circuits outweighs the improved eigenvalue resolution.

Taken together, the two panels validate the premise of our approach: repeated
application of a low-precision quantum oracle, combined with classical residual
updates, yields a progressively more accurate classical solution. The scheme
recovers the theoretically predicted dependence on oracle precision in the
noise-free setting, and remains convergent on hardware even when each individual
oracle call is corrupted by noise.

\subsection{Sensitivity analysis}

To directly quantify the resource savings afforded by IR, we compare the cumulative circuit depth required to solve a random $8$-variable linear system to a fixed residual threshold of $\|Ax - b\|_2 \leq 10^{-6}$ in noiseless simulation, with and without iterative refinement, as the condition number~$\kappa_A$ varies from~$10$ to~$50$.

Without iterative refinement, the cumulative circuit depth is simply the depth of a single high-precision HHL circuit. The precision of the algorithm is improved through increasing QPE register size (as the number of qubits in the QPE register corresponds to the number of digits of precision to which the eigenvalues of $A$ are estimated).  Note that while the precision of HHL improves exponentially with the number of QPE qubits $q$, this comes at the cost of exponentially increasing circuit depth, such that the HHL circuit depth scales linearly with the inverse precision.

With iterative refinement, the cumulative circuit depth corresponds to the depth of the (low-precision) HHL circuit used, multiplied by the number of iterative refinement iterations.

\begin{figure}[htbp]
    \centering
    \includegraphics[width=0.99\textwidth]{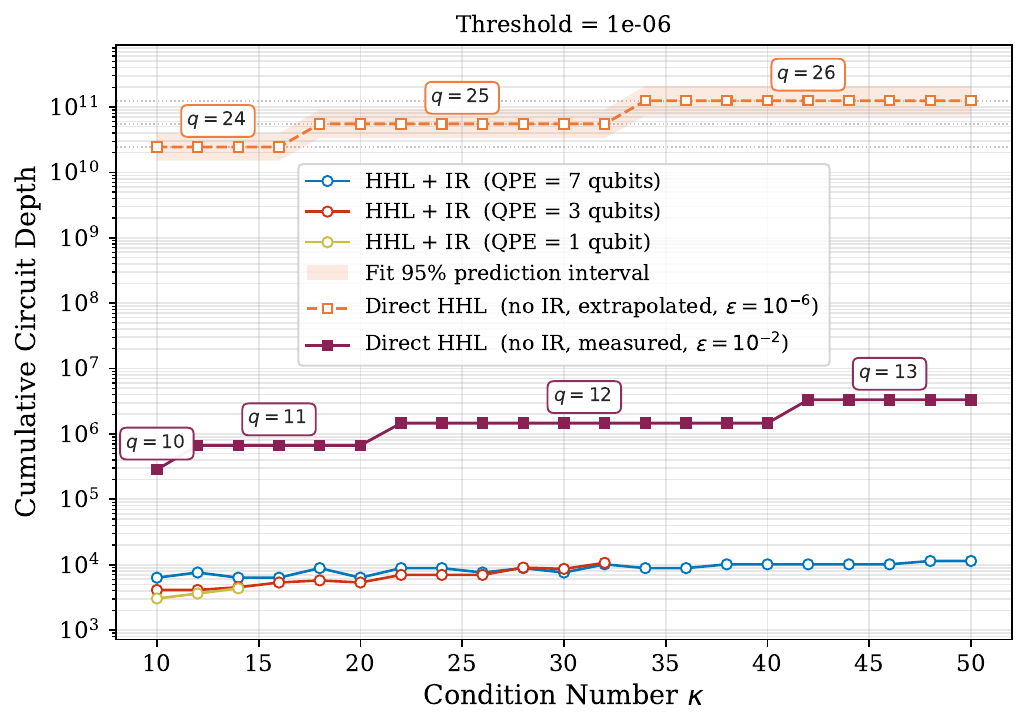}
    \caption{Cumulative circuit depth required to reach the residual threshold $\|Ax - b\|_2 \leq \varepsilon$, with and without iterative refinement.  Solid curves with filled markers are measured transpiled depths: HHL~+~IR at the target $\varepsilon = 10^{-6}$, and direct HHL at the relaxed threshold $\varepsilon = 10^{-2}$, for which the required QPE register sizes are directly compilable.  The dashed curve with hollow markers is the \emph{extrapolated} direct-HHL depth at $\varepsilon = 10^{-6}$, obtained from a fitted depth model; the shaded band is the fit's 95\% prediction interval.  Boxed labels indicate the required QPE register size~$q$ on each plateau.}
    \label{fig:sensitivity}
\end{figure}

For the IR approach, we fix a small number of QPE qubits ($q \in \{1,3,7\}$) and iterate Algorithm \ref{LS:alg:iterative refinement for LS} until the cumulative residual falls below the threshold.  The cumulative circuit depth is recorded directly from the transpiled circuits executed on the simulator at each iteration.

For the direct (no-IR) baseline, a single HHL call must resolve eigenvalues to precision $O(\varepsilon / \kappa_A)$, requiring $q \geq \lceil \log_2(\kappa_A / \varepsilon) \rceil$ QPE qubits.  We report two direct-HHL curves in Figure~\ref{fig:sensitivity}, carefully distinguishing measured from estimated quantities.  First, at a relaxed threshold of $\varepsilon = 10^{-2}$, the required register sizes ($q = 10$--$13$ over the tested range of~$\kappa_A$) remain directly compilable; this curve reports actual transpiled circuit depths with no modeling involved.  Second, at the target threshold of $\varepsilon = 10^{-6}$, the requirement grows to $q = 24$--$26$, requiring circuits far too deep to simulate or execute, and impractical even to transpile.  To estimate the corresponding depth, we transpile the full HHL circuit at every feasible QPE register size $q = 5, \ldots, 17$ and fit the model $\mathrm{depth}(q) \approx c_0 + c \cdot \gamma^{q}$ to the measured transpiled depths, obtaining a growth rate $\gamma \approx 2.26 \pm 0.03$ with $c \approx 79$. The constant $c_0$ has no physical interpretation: it merely allows the fit to account for contributions to the transpiled depth that grow more slowly than $\gamma^q$ (state preparation, readout, and the QFT stage of phase estimation), and its fitted value is strongly correlated with $c$ and $\gamma$. Its only role is to prevent these sub-exponential contributions from biasing the estimate of $\gamma$, which is the parameter that governs the extrapolation.  The fit is tight in log space ($R^2 = 0.999$, residual standard error $0.137$ in log-depth, i.e.\ a multiplicative depth uncertainty of roughly $\times/\div\, 1.15$ within the calibrated range); the shaded band in Figure~\ref{fig:sensitivity} shows the resulting 95\% prediction interval, which widens to roughly $\times/\div\, 1.7$ at the extrapolated register sizes.  The exponential scaling of transpiled depth with QPE register size is expected, as each additional QPE qubit roughly doubles the number of controlled rotations in the phase estimation circuit; the fitted rate slightly exceeds~$2$ because of the overhead from the eigenvalue inversion subroutine.

Figure~\ref{fig:sensitivity} displays the results.  The cumulative circuit depth for HHL~+~IR remains near $10^4$ across the full range of~$\kappa_A$ and is effectively independent of the condition number.  The comparison against direct HHL proceeds in two stages.  At the compilable threshold of $\varepsilon = 10^{-2}$, the measured direct depths range from $2.8 \times 10^5$ to $3.3 \times 10^6$---already a factor of $30$--$300$ deeper than IR, even though IR is held to the far stricter target of $10^{-6}$.  Tightening the direct baseline to the target $\varepsilon = 10^{-6}$ costs IR little more, while the extrapolated direct depth grows to $2.4 \times 10^{10}$--$1.3 \times 10^{11}$, a gap of six to seven orders of magnitude.  Moreover, both direct curves exhibit a staircase structure: each time $\kappa_A$ increases enough to require an additional QPE qubit, the circuit depth jumps by a factor of~$\gamma \approx 2.3$.  The IR curves display no such dependence, as the per-iteration circuit is fixed and the number of IR iterations grows only logarithmically in $1/\varepsilon$.

We also observe that the IR curves with $q = 1$ and $q = 3$ QPE qubits terminate at moderate values of~$\kappa_A$, as the per-iteration QLSA precision becomes insufficient to converge within the allotted iteration budget.  This reflects the practical requirement that the fixed precision~$\xi$ must satisfy $\xi < 1$ for geometric convergence (Theorem \ref{LS:theo; II-LS}); with too few QPE qubits, the effective per-iteration error improves too slowly to reach $\|Ax - b\|_2 \leq 10^{-6}$ in 30 iterations.  Nevertheless, even modest QPE precision ($q = 7$) suffices across the full range of tested condition numbers.

\subsection{Resource trade-offs}

A key advantage of the IR scheme is the ability to trade iteration count for circuit complexity.  Figure~\ref{fig:ir_tradeoffs} illustrates this trade-off for a $64$-variable linear system solved in simulation.  The left panel fixes the number of QPE qubits at $q = 6$ and varies the number of tomography shots per IR iteration; the right panel fixes the shot count at $10^6$ and varies the number of QPE qubits.  Color encodes the resulting residual $\log_{10}(\|Ax^{(k)} - b\|_2)$ at each parameter combination, and dashed contour lines show the cumulative quantum resources (shots and circuit depth, respectively) across all IR iterations up to that point.

\begin{figure}[htbp]
    \centering
    \begin{minipage}[b]{0.49\textwidth}
        \centering
        \includegraphics[width=\textwidth]{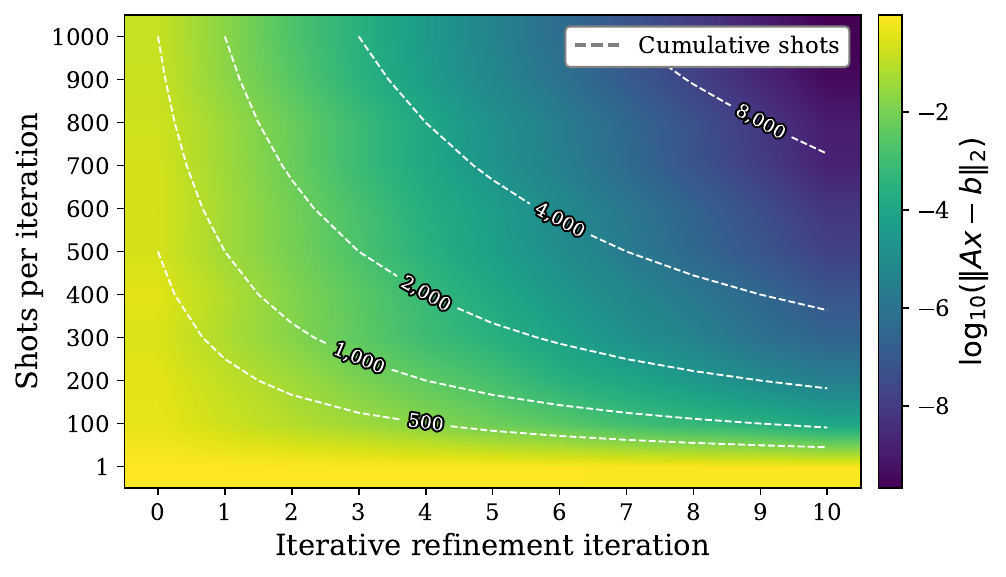}
        \label{fig:shots_vs_ir}
    \end{minipage}
    \begin{minipage}[b]{0.49\textwidth}
        \centering
        \includegraphics[width=\textwidth]{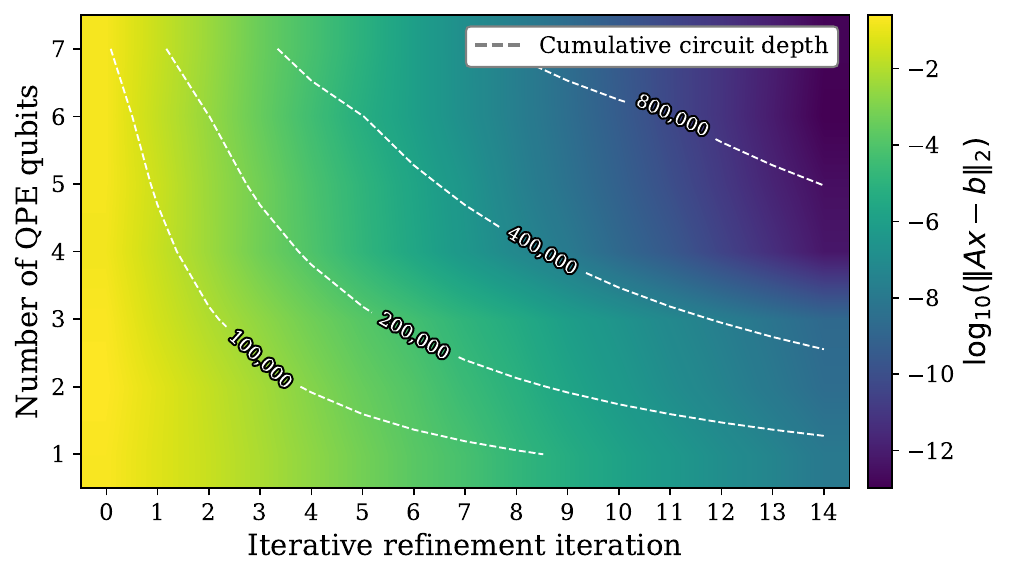}
        \label{fig:qpe_vs_ir}
    \end{minipage}
    \caption{Trade-offs between iterative refinement iterations and quantum resources for solving a 64-variable linear system to a desired residual using HHL. Color encodes the resulting residual at each parameter combination as $\log_{10}(\|Ax - b\|_2)$, and contour lines show cumulative quantum resources across all IR iterations.  In the left plot, 6 QPE qubits are used for the HHL circuit, and in the right plot, $10^6$ shots are used to readout x from $\ket{x}$ at each iteration of iterative refinement.}
    \label{fig:ir_tradeoffs}
\end{figure}
 
These heatmaps reveal the characteristic structure of the trade-off: along the horizontal axis, each additional IR iteration incrementally improves precision at constant per-iteration cost; along the vertical
axis, increasing the QPE register size (or shot budget) yields faster convergence per iteration but at the expense of deeper circuits (or more total measurements).  The contour lines make the iso-resource surfaces
explicit.  For instance, achieving a residual of $10^{-3}$ with 100 shots per iteration requires approximately 10 IR iterations, whereas 1000 shots per iteration achieves the same residual in only 2 iterations---but with roughly double the cumulative shots (2000 versus 1000). Operationally, given a fixed budget of quantum resources (shot count and/or circuit depth), one should aim to minimize the residual achieved along the corresponding contour line; in every case, this optimum is attained by using $k > 1$ iterations of iterative refinement rather than a single high-precision call.

The same principle applies to circuit depth, which is the dominant resource constraint on near-term hardware.  The right panel of Figure~\ref{fig:ir_tradeoffs} shows that one can reach arbitrarily high precision $\zeta$ by performing $O(\log(1/\zeta))$ iterations of a fixed, shallow circuit, rather than constructing a single deep circuit whose depth scales polynomially in~$1/\zeta$.  The per-iteration circuit depth grows approximately
exponentially in~$q$ (roughly a factor of 2 per additional QPE qubit), reflecting the cost of implementing controlled rotations of increasing precision in the QPE subroutine.  The IR scheme fixes $q$ at a modest
value and compensates with additional iterations, keeping the per-iteration depth constant and the total depth scaling linearly in the iteration count.

\subsection{Discussion}
 
Taken together, these experiments corroborate the theoretical predictions of Section~\ref{sec:IR}. The geometric convergence rate of Theorem~\ref{LS:theo; II-LS} is clearly visible in the simulator results for well-conditioned systems, where each iteration reduces the residual by a roughly constant factor on the logarithmic scale. The hardware experiment on \texttt{ibm\_pittsburgh} demonstrates that this convergence persists---albeit at a reduced rate---on noisy superconducting devices, lending practical credibility to the IR framework.
 
The sensitivity analysis highlights that the effective convergence rate depends on the interplay between $\kappa_A$ and the QPE precision $2^{-q}$ on HHL circuits: when $q$ is large enough to resolve the eigenvalue spectrum of $A$, convergence is rapid; when it is not, additional QPE qubits are needed. Crucially, even in regimes where a single QPE register cannot achieve the target precision outright, the IR scheme allows one to distribute the precision requirement across multiple iterations, each employing a manageable circuit depth.
 
From a resource perspective, the trade-off heatmaps confirm that the IR approach offers a practical path to high-precision solutions on near-term hardware: rather than requiring circuits whose depth grows polynomially in $1/\zeta$, one can execute $O(\log(1/\zeta))$ repetitions of a constant-depth circuit, paying only a logarithmic overhead in the total number of quantum operations. This makes the scheme particularly attractive for early fault-tolerant devices, where circuit depth is still a primary bottleneck.

\section{Conclusion}\label{s:con}
In this work, we present an iterative refinement scheme that uses low-precision quantum oracles to accurately solve classical linear systems. The key advantage of the proposed approach is that all quantum subroutines operate at fixed precision, reducing the dependence on the final target precision from polynomial to polylogarithmic when QLSAs are combined with quantum state tomography.

We analyze the framework under QRAM, sparse-access, and LCU input models, and support the theoretical results with numerical experiments on both simulators and quantum hardware. The experiments demonstrate geometric residual reduction and show that repeated low-precision quantum solves can be more resource-efficient than a single high-precision solve.

The main remaining classical cost is the matrix-vector multiplication required to compute the residual at each refinement step. Developing efficient quantum approaches for residual computation is a natural direction for future research.


\section*{Acknowledgment}
An earlier technical report \cite{mohammadisiaroudi2023exponentially}, related to this work, was posted in 2023.
We thank Ramin Fakhimi for his contributions to an earlier implementation of the QLSA and IR framework.
This work is supported by the Defense Advanced Research Projects Agency as part of the project W911NF2010022: {\em The Quantum Computing Revolution and Optimization: Challenges and Opportunities}. G.~Nannicini is supported by ONR award \# N000142312585. A.~Harkness is supported by PNNL’s Quantum Algorithms and Architecture for Domain Science (QuAADS) Laboratory Directed Research and Development (LDRD) Initiative. This material is based upon work supported by the U.S. Department of Energy, Office of Science, National Quantum Information Science Research Centers, Quantum Science Center (QSC). The Pacific Northwest National Laboratory is operated by Battelle for the U.S. Department of Energy under Contract DE-AC05-76RL01830. This research used resources of the Oak Ridge Leadership Computing Facility (OLCF), which is a DOE Office of Science User Facility supported under Contract DE-AC05-00OR22725. This research used resources of the National Energy Research Scientific Computing Center (NERSC), a U.S. Department of Energy Office of Science User Facility located at Lawrence Berkeley National Laboratory, operated under Contract No. DE-AC02-05CH11231.

{\hyphenpenalty=100000
\bibliographystyle{dependencies/siamplain}
\bibliography{dependencies/references}}
\end{document}